\documentclass[leqno]{article}
\usepackage{amsfonts,amsmath,amsthm,amssymb,times}
\usepackage{graphicx}
\usepackage{hyperref}
\usepackage{enumerate}
\usepackage{color}
\usepackage{rotating}
\newtheorem{theorem}{THEOREM}[section]
\newtheorem{lemma}{LEMMA} [section]
\newtheorem{definition}{DEFINITION}[section]
\newtheorem{corollary}{COROLLARY}[section]

\newtheorem{proposition}{PROPOSITION}[section]

 \theoremstyle{definition}
\numberwithin{equation}{section}

\begin{document}

\title{ \bf 3D mean Projective Shape Difference for Face Differentiation from Multiple Digital Camera Images }

\author{ K. D. Yao \thanks{Research
supported by National Science Foundation Grant DMS-1106935 and the
National Security Agency Grant MSP-
H98230-15-1-0135}, V. Patrangenaru \thanks{Research
supported by National Science Foundation Grant DMS-1106935 and the
National Security Agency Grant MSP-
H98230-15-1-0135} and D. Lester \thanks{Research
supported by National Science Foundation Grant DMS-1106935 and the
National Security Agency Grant MSP-
H98230-15-1-0135} \\ Florida
State University}

 \maketitle
\begin{abstract}
We give a nonparametric methodology for hypothesis testing for equality of
extrinsic mean objects on a manifold embedded in a numerical spaces. The
results obtained in the general setting are detailed further in the case of
3D projective shapes represented in a space of symmetric matrices via the
quadratic Veronese-Whitney (VW) embedding. Large sample and nonparametric
bootstrap confidence regions are derived for the common VW-mean of random projective shapes
for finite 3D configurations. As an example, the VW MANOVA testing methodology is
applied to the multi-sample mean problem for independent projective shapes of $3D$ facial configurations retrieved
from digital images, via Agisoft PhotoScan technology.
\end{abstract}

\section{Introduction}

\noindent In this paper, we continue the Object Data Analysis program started by  Patrangenaru and Ellingson (2015)\cite{PaEl:2015}. In Section \ref{sc2} we revisit the hypothesis testing for equality of mean vectors from $g$ multivariate populations, in  nonparametric setting based on the idea that the numbers in a finite set are all equal, if the squares of their differences add up to zero (see Bhattacharya and Bhattacharya(2012)\cite{BhBh:2012}. The main difference between our approach and classical MANOVA, is that we do not assume that all populations have a common covariance matrix $\Sigma,$ and we do not make any distributional assumption. In Section \ref{sc3}, we extend this methodology to test for the equality of multiple extrinsic means, based on random samples of various sizes collected from $g$ independent probability measures on a manifold. Our newly developed extrinsic MANOVA test is applied to the particular case of $3$D projective shape data in section \ref{sc4}, using the Veronese Whitney embedding of the projective shape space (see eg. Mardia and Patrangenaru(2005)\cite{MaPa:2005}. This method builds upon previous results on one sample hypothesis testing methods, as developed in Patrangenaru et al. (2010\cite{PaLiSu:2010}, 2014\cite{PaQiBu:2014}). The space $P\Sigma_{3}^{k}$  of 3D projective shapes of $k$-ads including a projective frame at given landmark indices is isomorphic to $(\mathbb{R}P^3)^{k-5}$. Therefore a 3D projective shape face differentiation via VW-MANOVA testing is presented in Section \ref{sc5}. Note that behind the 3D Agisoft reconstruction software are results by Faugeras(1992)\cite{Fa:1992} and Hartley et. al.(1992)\cite{HaGuCh:1992}, showing that a 3D configuration of landmarks can be obtained from multiple noncalibrated camera images up to a projective transformation in 3D, thus allowing us to conduct without ambiguity a 3D projective shape analysis.

\section{Motivations for new MANOVA on manifolds}\label{sc2}
\noindent For $a=1,...,g,$ suppose $X_{a,i}\sim N_p(\mu_a,\Sigma), i=1,...,n_a$ are $p$ dimensional i.i.d random vectors. To test if the mean vectors of the $g$ groups are the same, one considers the hypothesis testing problem
\begin{align}
H_0 &:\ \mu_1=\mu_2=...=\mu_g=\mu  \label{MANOVA_null}\\
H_a &:\ at\ least\ one\ equation\ does\ not\ hold. \notag
\end{align}
Assuming that the covariance matrix $\Sigma$ is invertible, by the Central Limit Theorem, for large sample sizes $n_a, a = 1, \dots, g,$ we have
\begin{align}
\sqrt{n_a}\Sigma^{-\frac{1}{2}}(\bar{X}_a-\mu) & \sim N_p(0_p,I_p),\\
n_a (\bar{X}_a-\mu)^T \Sigma^{-1}(\bar{X}_a-\mu)& \sim \chi_p^2.
\end{align}
However, $\Sigma$ is always unknown, so in practice, one has to use its unbiased estimator $S_a,\ a=1,...,g.$
\begin{equation}
n_a (\bar{X}_a-\mu)^T S_a^{-1}(\bar{X}_a-\mu)\sim \chi_p^2.
\end{equation}
Let us consider the pooled sample mean $\bar{X}=\frac{1}{n}(n_1\bar{X}_1+...+n_g\bar{X}_g),\ n=\sum_{a=1}^{g}n_a.$
\begin{lemma}\label{L1}
Under the null, $\bar{X}$ is a consistent estimator of $\mu,$ provided $\frac{n_a}{n}\rightarrow \lambda_a>0,\ as\ n\rightarrow \infty, \ a=1,...,g$.
\end{lemma}
\begin{proof}
Indeed, for any $a\in \{1,2,...,g\}$, since $\frac{n_a}{n}\rightarrow \lambda_a>0,\ as\ n\rightarrow \infty$, and $\bar{X}_a$ is the consistent estimator of $\mu$, therefore,
\begin{equation}
\bar{X}\rightarrow_p \lambda_1 \mu+\lambda_2 \mu+...+\lambda_g \mu=\mu.
\end{equation}
\end{proof}

\begin{theorem}
The statistic for the hypothesis in \eqref{MANOVA_null} is
\begin{equation}
\sum_{a=1}^g n_a (\bar{X}_a-\bar{X})^T S_a^{-1}(\bar{X}_a-\bar{X})\sim \chi_{gp}^2.
\end{equation}
So the rejection region for the test is
\begin{equation}\label{eq:manova}
\sum_{a=1}^g n_a (\bar{X}_a-\bar{X})^T S_a^{-1}(\bar{X}_a-\bar{X})>\chi_{gp}^2(c).
\end{equation}
\end{theorem}

\section{ MANOVA on manifolds }\label{sc3}
\noindent In this section we will focus on the asymptotic behavior of statistics related to means on a  manifold $\mathcal{M}$ based on samples of different sizes from different populations on $\mathcal{M}.$
Now let's consider the set  $X_{a,1}, \ldots,X_{a,n_a}$ ($a=1,2,...,g$) of iid random objects on $\mathcal{M}$ with common probability measure $Q_{a}.$ We denote the extrinsic mean of the $j$- nonfocal probability measure $Q_a$ on $\mathcal{M}$ by $\mu_{a,E}$ for ease of notation and because there is no ambiguity about the embedding used.  The corresponding extrinsic sample means are written $\bar{X}_{a,E}$ for $a=1, \cdots, g.$ From this point on, we will assume that all the distributions are $j$-nonfocal.

\subsection{Hypothesis testing and $T^2$ statistic}\label{ssc3.1}

Assume $X_{a,1}, \ldots,X_{a,n_a}$ are iid random objects on $\mathcal{M}$  a $p$-dimensional  manifold, with probability measure $Q_{a}$ with $a=1,2,...,g$. We are interested in comparing multiple extrinsic means.

\noindent We would like to develop a test similar to \eqref{eq:manova} designed to test the difference between the $g$ extrinsic means. One challenge that presents itself at the early stage is a proper definition of  a pooled mean for random objects on a $p$-dimensional manifold $\mathcal M.$ Linearity becomes an issue when dealing with extrinsic means. For a proper definition we will focus on the equalities tied to the assumption
$$A_0: \mu_{1,E}= \cdots = \mu_{g, E}$$

\begin{definition}\label{d31}
Under the assumption $A_0$ and  for any $a\in \{1,2,...,g\}$, with $\frac{n_a}{n}\rightarrow \lambda_a>0,\ as\ n\rightarrow \infty.$ We define
\begin{enumerate}[(i)]
\item The {\bf pooled extrinsic mean with weights $\lambda = (\lambda_1, \dots, \lambda_g)$,} denoted $\mu_{E}(\lambda)$ as the value in $\mathcal M$ given by
\begin{equation}
j(\mu_{E})= P_j(\lambda_1 j(\mu_{1,E}) + \cdots + \lambda_{g}j(\mu_{g,E})) \label{E_Polled_mean}
\end{equation}
Where $\mu_{a,E}$ is the extrinsic mean of the random object $X_{a,1}$ and $\Sigma_{a=1}^{g} \lambda_a= 1$
\item The {\bf extrinsic pooled sample mean } denoted $\bar{X}_{E} ~\in~\mathcal M$ given by;
{
\begin{equation}
j(\bar{X}_{E})= P_j \left(\frac{n_1}{n} j(\bar{X}_{1, E}) + \cdots + \frac{n_g}{n}j(\bar{X}_{g, E}) \right) \label{E_Polled_sample_mean}
\end{equation}
Where $\bar{X}_{a, E}$ is the extrinsic sample mean for $X_{a,1}$ and $n=\sum_{a=1}^{g} n_a $ }
\end{enumerate}
\end{definition}
\smallskip

\noindent Note that since $A_0$ implies $j(\mu_{1,E})= \cdots = j(\mu_{g, E}),$ and with our definition of the extrinsic pooled mean we get $j(\mu_{E})= j(\mu_{a,E})$ for each $a=1,\dots, g.$
 Furthermore, the linear combination $ \lambda_1 j(\mu_{1,E}) + \cdots + \lambda_{g}j(\mu_{g,E}) \in j(\mathcal M).$  Note that for $a=1, \cdots , g$ $\bar{X}_{a, E}$  is a consistent estimator of $ \mu_{a,E}$ and therefore we get that $j(\bar{X}_{E}) \to_{p} j(\mu_{E}).$ Since $j$ is a homeomorphism from $\mathcal M$ to $j(\mathcal M)$ we also have that $\bar{X}_{E}$ is a consistent estimator of $\mu_{E}$ the extrinsic pooled mean.
 With this definition at hand, we now express the following hypothesis test, designed to test the difference between extrinsic means and is given by;

\begin{align}
\label{MANOVA_null_ex}
H_0 &:\ \mu_{1,E}=\mu_{2,E}=...=\mu_{g,E}= \mu_E,\\
H_a &:\ at\ least\ one\ equality\ \mu_{a,E}=\mu_{b,E}, 1 \leq a < b\leq g \ does\ not\ hold. \notag
\end{align}

And  since the embedding $j: \mathcal M \to \mathbb R^N$ is one-to-one the hypothesis above can be interchangeably written
\begin{align}
\label{MANOVA_jnull}
H_{0}^{j} &:\  j(\mu_{1,E})=j(\mu_{2,E})=...=j(\mu_{g,E})= j(\mu_E),\\
H_{a}^{j} &:\ at\ least\ one\ equality\ \mu_{a,E}=\mu_{b,E}, 1 \leq a < b\leq g \ does\ not\ hold. \notag
\end{align}

\noindent In order to test hypothesis \eqref{MANOVA_null_ex} we will use a $T^2$ like statistic. The theorem below, gives us the asymptotic behavior needed to establish such a statistic. For $a=1,\dots,g,$ we get, from Bhattacharya and Patrangenaru \cite{BhPa:2005}, the following:
\begin{enumerate}[(i)]
\item $\displaystyle{S_{n_a}= (n_a )^{-1} \Sigma_{i=1}^{n_a}(j(X_{a,i})- j(\bar{X}_{E})) (j(X_{a,i})- j(\bar{X}_{E}))^T   }$ is a consistent estimator of $\Sigma_{a},$ the covariance matrix of $X_{a,1}$ and
\item $\tan_{j(\bar{X}_{E}) } \nu$ is a consistent estimator of $\tan_{P_{j}(\mu)} \nu,$ where $\nu \in \mathbb{R}^{N}.$
\end{enumerate}
 It follows that, under \eqref{MANOVA_jnull}, $S_{E,a}(j,X_a),$ given by
 {
\begin{align}
S_{E,a}(j,X_a)&=\left[ \left[\sum_{a=1}^{m} d_{\overline{j^{(p)}(X)}}P_{j}(e_b) \cdot e_{i}(j(\bar{X}_{E}))  ~e_{i}(j(\bar{X}_{E}))\right]_{i=1,...,p} \right] \cdot ~S_{n_a}\notag\\
&~~~~\left[ \left[\sum_{a=1}^{m} d_{\overline{j^{(p)}(X)}}P_{j}(e_b) \cdot e_{i}(j(\bar{X}_{E})) e_{i}(j(\bar{X}_{E})) \right]_{i=1,...,p} \right]^T \notag
\end{align}
where for $\overline{j^{(p)}(X)}=\frac{n_1}{n} j(\bar{X}_{1, E}) + \cdots + \frac{n_g}{n}j(\bar{X}_{g, E})$ and $P_j(\overline{j^{(p)}(X)})$ is a  consistent estimator of  $j(\mu_E).$ One must note that the extrinsic sample covariance matrix $S_{E,a}(j,X_a)$ is expressed in terms of $d_{\overline{j^{(p)}(X)}}P_{j}(e_b) \in T_{j(\bar{X}_{E})} j(\mathcal M)$ and not in term of $d_{\overline{j(X{a,1})}}P_{j}(e_b) \in T_{j(\bar{X}_{a,E})} j(\mathcal M).$ 

\bigskip

\begin{theorem}\label{Theorem31}
Assume $j: \mathcal{M} \rightarrow \mathbb{R}^N$ is a closed embedding of $\mathcal{M}$. Let $\{ X_{a,i} \}_{i=1}^{n_a}$ for $a=1,...,g$ be random samples from the $j$-nonfocal distributions $\mathcal{Q}_a$. Let $\mu_a = E(j(X_{a,1}))$ and assume $j(X_{a,1})$'s have finite second-order moments and the extrinsic covariance matrices $\Sigma_{a,E}$ of $X_{a,1}$ are nonsingular. We also let $\left( e_{1}(p), ...., e_{N}(p) \right)$, for $p~\in \mathcal{M}$ be an orthonormal frame field adapted to $j$.\\
For $a =1,\dots, g,$ assume $\lambda_{a} >0$ are constants, such that $\sum_{a=1}^{g} \lambda_a =1$. Furthermore, let $\frac{n_a}{n} \rightarrow \lambda_{a} >0$, as $n\rightarrow \infty$, with $n=\Sigma_{a=1}^{g} n_a.$ Then  we have the following asymptotic behavior;
\begin{equation}
\sum_{a=1}^g n_a~ \tan_{j(\mu_E)}(j(\bar{X}_{a,E})-j(\mu_E))^T \Sigma_{a,E}^{-1}~\tan_{j(\mu_E)} (j(\bar{X}_{a,E})-j(\mu_E))  \to_d   \chi_{gp}^2. \notag
\end{equation}
\noindent It follows that the statistics for hypothesis \eqref{MANOVA_null_ex} have the following asymptotic results;
\begin{enumerate}[(a)]
\item the statistic \begin{align}
\sum_{a=1}^g n_a~ \tan_{j(\mu_E) } (j(\bar{X}_{a,E})-j(\bar{X}_E))^T S_{E,a}(j,X_a)^{-1}~\tan_{j(\mu_E) } (j(\bar{X}_{a,E})-j(\bar{X}_E))   \to_d    \chi_{gp}^2.\notag
\end{align}
\item the statistic\begin{align}
\sum_{a=1}^g n_a~ \tan_{j(\bar{X}_{E})} (j(\bar{X}_{a,E})-j(\bar{X}_E))^T S_{E,a}(j,X_a)^{-1}~\tan_{j(\bar{X}_E)} (j(\bar{X}_{a,E})-j(\bar{X}_E))   \to_d    \chi_{gp}^2. \notag
\end{align}

\end{enumerate}

\end{theorem}
\begin{proof}

Recall from the Bhattacharya and Patrangenaru(2005)\cite{BhPa:2005}, from the consistency of the sample mean vector and from the continuity of the
projection map $P_j,$ that we have
\begin{align}
&\sqrt{n_a}~\tan_{j(\mu_{E})}(j(\bar{X}_{a,E})-j(\mu_{E})) \rightarrow_d N(0_p,\Sigma_{a,E}), ~~~for~ a=1,2,...,g \notag\\
\text{where} & \notag\\
\Sigma_{a,E}&=\left[ \left[  \sum_{b=1}^N d_{\mu}P_{j}(e_b) \cdot e_{k}(P_{j}(\mu)) \right]_{k=1,...,p}\right]~~\Sigma_{a}~~\left[\left[ sum_{b=1}^N d_{\mu}P_{j}(e_b) \cdot e_{k}(P_{j}(\mu))  \right]^T_{k=1,...,p} \right], \notag
\end{align}
where $\mu= \lambda_1 j(\mu_{1,E}) + \cdots + \lambda_{g}j(\mu_{g,E})$ and $\Sigma_{a}$ is the covariance matrices of the $j(X_{a,1})$ with respect to the canonical basis $e_1,...,e_{N}.$ And under the null, from \eqref{MANOVA_null_ex}, the matrices $\Sigma_{a,E}$ are defined with respect to the basis $f_1(\mu_{E}),..., f_{p}(\mu_{E})$ of local frame fields, $f_r=d_{j^{-1}(P_j(\mu))}(e_r(P_j(\mu))).$  We then have for each $a=1,...,g$
$$ n_a \tan_{j(\mu_E)}(j(\bar{X}_{a,E})-j(\mu_E))^T\Sigma_{a,E}^{-1} \tan_{j(\mu_E)}(j(\bar{X}_{a,E})-j(\mu_E))  \to_d   \chi_{p}^2.$$ and since the random samples are independent we have,
\begin{equation}
\sum_{a=1}^g n_a \tan_{j(\mu_E)}(j(\bar{X}_{a,E})-j(\mu_E))^T\Sigma_{a,E}^{-1}~\tan_{j(\mu_E)}(j(\bar{X}_{a,E})-j(\mu_E))  \to_d   \chi_{gp}^2.
\end{equation}
$\bar{X}_{E}$ is the consistent estimator of $\mu_E$, then the pooled sample mean
\begin{align}
j(\bar{X}_{E})= P_{j} \left( \frac{1}{n} \sum_{a=1}^{g} n_a j(\bar{X}_{a,E})\right)  \rightarrow_{p} j(\mu_{E}) ~~~(\text{by Lemma}~ \ref{L1})
 \end{align}
 {
And since $S_{E,a}(j,X_a)$ consistently estimates $\Sigma_{a}$ and $\tan_{j(\bar{X}_{E})}$ is a consistent estimator of $\tan_{j(\mu_E)}$, we have the following
\begin{align}
 \sum_{a=1}^g n_a \tan_{j(\mu_E)}(j(\bar{X}_{a,E})-j(\bar{X}_E))^T S_{E,a}(j,X_a)^{-1}~\tan_{j(\mu_E)}(j(\bar{X}_{a,E})-j(\bar{X}_E)) \to_d   \chi_{gp}^2. \notag\\
\sum_{a=1}^g n_a~ \tan_{j(\bar{X}_{E})} (j(\bar{X}_{a,E})-j(\bar{X}_E))^T S_{E,a}(j,X_a)^{-1}~\tan_{j(\bar{X}_{E}) } (j(\bar{X}_{a,E})-j(\bar{X}_E))   \to_d    \chi_{gp}^2. \notag
\end{align}
}
\end{proof}
\subsection{Nonparametric bootstrap confidence regions for the common extrinsic mean}\label{ssc3.2}
\noindent From Bhattacharya and Patrangenaru(2005)\cite{BhPa:2005} and from Corollary 3.2 in
Bhattacharya and Bhattacharya(2012)\cite{BhBh:2012}, under the hypothesis \\
$~~~\quad \quad \quad ~~~\begin{cases}
H_0 &:\ \mu_{1,E}=\mu_{2,E}=...=\mu_{g,E}= \mu_E, \notag\\
H_a &:\ni (i,j) 1\le i<j <g, \text{s.t.} \ \mu_{i,E} \ne \mu_{j,E} , \notag
\end{cases}$\\
we have:\\
\begin{corollary}\label{c}
Under the assumptions of Theorem \eqref{Theorem31} , a confidence  regions for $\mu_E$ of asymptotic level $1- c$ is given by $C_{n, c}^{(g)}$ and $D_{n, c}^{(g)}$ which are defined below
\begin{enumerate}[(a)]
\item $C_{n, c}^{(g)}= j^{-1}(U_{n , c})$ where \\$U_{n, c}= \{  j (\nu) \in j(\mathcal{M}): \sum_{a=1}^g n_a  \left\| S_{E,a}(j,X_{a})^{-1/2}~\tan_{ j (\nu)}(j(\overline{X}_{a,E})- j (\nu) )\right\|^2 \leq \chi^{2}_{gp, 1-c} \}$
\item $D_{n, c}^{(g)}= j^{-1}(V_{n , c})$ where \\$V_{n, c}= \{  j(\nu) \in j(\mathcal{M}):  \sum_{a=1}^g n_a \left\| S_{E,a}(j,X_{a})^{-1/2}~\tan_{j(\bar{X}_{E})}(j(\overline{X}_{a,E})-j(\nu) )\right\|^2 \leq \chi^{2}_{gp, 1-c} \}$
\end{enumerate}
where $\bar{X}_{E}$ is the extrinsic pooled sample mean defined in Definition \ref{d31} $(ii)$
\end{corollary}
\noindent Most of the data we will be focusing on will have value of $n$ relatively small. We will need to use resampling, in particular, bootstrap methods.  For $a=1,...,g,$ let $~\{ X_{a,i} \}_{i=1}^{n_a}$ be i.i.d.r.o's from the $j$-nonfocal distributions $\mathcal{Q}_a.$ Let $\{X_{a,r}^{*} \}_{r=1,...,n_a}$ be random resamples with repetition  from the empirical $\hat{Q}_{n_a}$ conditionally given $\{ X_{a,i} \}_{i=1}^{n_a}.$
\noindent The confidence regions $C_{n, c}^{(g)}$ and $D_{n, c}^{(g)}$ described above  have the corresponding bootstrap analogue ${C^*}_{n, c}^{(g)}$ and ${D^*}_{n, c}^{(g)}$ which are defined in the corollary below.

\begin{corollary}\label{c2}
The $(1-c) 100 \%$ bootstrap confidence regions for $\mu_{E}$ with $d=g p$ are given by
\begin{enumerate}[(a)]
\item ${C^*}_{n, c}^{(g)}= j^{-1}(U^{*}_{n , c})$ and
\begin{equation}\label{cstar_g}
U^{*}_{n , c}= \{  j(\nu) \in j(\mathcal{M}):   \sum_{a=1}^g n_a \left\| S_{E,a}(j,X_{a})^{-1/2}~\tan_{j(\nu)}(j(\overline{X}_{a,E})-j(\nu))\right\|^2 \leq {c^*}^{(g)}_{1-c} \}
\end{equation}
\noindent where ${c^*}^{(g)}_{1-c}$ is the upper $100(1-c) \%$ point of the values
\begin{equation}
 \sum_{a=1}^g n_a \left\|  S_{E,a}(j,X^{*}_{a})^{-1/2}~\tan_{j(\bar{X}_{E})}(j(\overline{X^*}_{a,E})-j(\bar{X}_{E}))\right\|^2
\end{equation} among the bootstrap re samples.
\item ${D^*}_{n, c}^{(g)}= j^{-1}({V^*}_{n , c})$ and
\begin{equation}\label{dstar_g}
{V^*}_{n, c}= \{ j(\nu ) \in j(\mathcal{M}):  \sum_{a=1}^g n_a \left\| S_{E,a}(j,X_{a})^{-1/2}~\tan_{j(\bar{X}_{E})}(j(\overline{X}_{a,E})-j(\nu) )\right\|^2  \leq {d^*}^{(g)}_{1-c} \}
\end{equation}
where ${d^*}^{(g)}_{1-c} $ is the upper $100(1-c)\%$ point of the values
\begin{equation}
\sum_{a=1}^g n_a \left\|  S_{E,a}(j,X^{*}_{a})^{-1/2}~\tan_{j({\bar{X}^*}_{E})}(j(\overline{X^*}_{a,E})-j(\bar{X}_{E}))\right\|^2
\end{equation}
\end{enumerate}
where ${\bar{X}^*}_{E}$ is the extrinsic pooled re sampled mean given by
\begin{equation}\label{E_Polled_sample_mean}
j({\bar{X}^*}_{E})= P_j \left( \frac{n_1}{n} j(\bar{X}^{*}_{1, E}) + \cdots + \frac{n_g}{n}j(\bar{X}^{*}_{g, E}) \right)
\end{equation}
among the bootstrap re samples. Both of the regions given by \eqref{dstar_g} and \eqref{cstar_g} have coverage erro $O_{p}(n^{-2}).$
\end{corollary}

\noindent Note that $ S_{E,a}(j,X^{*}_{a})$
\begin{align}
S_{E,a}(j,X^{*}_{a})&=\left[ \left[\sum_{a=1}^{m} d_{\overline{j^{(p)}(X^{*})}}P_{j}(e_b) \cdot e_{i}(j(\bar{X}^{*}_{E}))  ~e_{i}(j(\bar{X}^{*}_{E}))\right]_{i=1,...,p} \right] \cdot ~S^{*}_{n_a}\notag\\
&~~~~\left[ \left[\sum_{a=1}^{m} d_{\overline{j^{(p)}(X^{*})}}P_{j}(e_b) \cdot e_{i}(j(\bar{X}^{*}_{E})) e_{i}(j(\bar{X}^{*}_{E})) \right]_{i=1,...,p} \right]^T \notag
\end{align}
where $\displaystyle{S^{*}_{n_a}= (n_a )^{-1} \Sigma_{i=1}^{n_a}(j(X^{*}_{a,i})- j(\bar{X}^{*}_{E})) (j(X^{*}_{a,i})- j(\bar{X}^{*}_{E}))^T   }.$

\noindent We now express the following test statistics that will be used in our analysis and are tied to the confidence regions  mentioned above.

\begin{proposition}\label{t2}
 Let $\{ X_{a,i} \}_{i=1}^{n_a}$ for $a=1,...,g$ be random samples from the $j$-nonfocal distributions $\mathcal{Q}_a.$  Let $\mu_a = E(j(X_{a,1}))$ and assume $j(X_{a,1})$'s have finite second-order moments and the extrinsic covariance matrices
 $\Sigma_{a,E}$ of $X_{a,1}$ are nonsingular.
 \begin{enumerate}[(a)]
 \item Then the distribution of
 $T_{c}({X}^{(g)},\hat{Q}^{(g)})= \sum_{a=1}^g n_a  \left\| \Sigma_{a,E}^{-1/2}~\tan_{ j (\mu_E)}(j(\overline{X}_{a,E})- j (\mu_E) )\right\|^2$  can be approximated by the bootstrap distribution function of\\
 ${T_{c}(X^{*(g)},\hat{Q}^{(g)})}=\sum_{a=1}^g n_a \left\|  S_{E,a}(j,X^{*}_{a})^{-1/2}~\tan_{j(\bar{X}_{E})}(j(\overline{X}^{*}_{a,E})-j(\bar{X}_{E}))\right\|^2$
\item Similarly, the distribution of  $T_{d}({X}^{(g)},\hat{Q}^{(g)})=\sum_{a=1}^g n_a \left\| S_{E,a}^{-1/2}~\tan_{j(\bar{X}_{E})}(j(\overline{X}_{a,E})-j(\mu_E) )\right\|^2 $  can be approximated by the bootstrap distribution function of\\
 ${T_{d}(X^{*(g)},\hat{Q}^{*(g)})}=\sum_{a=1}^g n_a \left\|  S_{E,a}^{*}(j,{X}_{a}^{*})^{-1/2}~\tan_{j({\bar{X}^*}_{E})}(j(\overline{X^*}_{a,E})-j(\bar{X}_{E}))\right\|^2$
 \end{enumerate}
  with coverage error $O_p(n^{-2})$.
\end{proposition}

 \noindent Note that ${T(X^{*(g)},\hat{Q}^{(g)})}$ is obtained from ${T(X^{(g)},\hat{Q}^{(g)})}$ by substituting $X^{(g)}_1= (X_{1,1},\cdots, X_{g,1})^T$ with  re samples $X^{* (g)}_1= (X^{*}_{1,1},\cdots, X^{*}_{g,1})^T.$

\noindent Using the bootstrap analogue in the previous Proposition \ref{t2} yields  simpler method for finding  $100(1-c)\%$ confidence regions. We will utilize the tests statistics expressed above to conduct our analysis with confidence regions
 $C^{*}_{n,c}$  and $D^{*}_{n, c}$ as shown in the Corollary \ref{c2}.

\section{MANOVA on $(\mathbb R P^{3})^q$}\label{sc4}
\noindent We start with the 3-dimensional real projective space $\mathbb RP^3,$ set of 1-dimensional linear subspaces of $\mathbb R^{4}.$ $\mathbb RP^3$ has a 3D manifold structure (see Patrangenaru and Ellingson(2015)\cite{PaEl:2015},p.106). A point $p=[x] \in \mathbb R P^3$, is the equivalence class of $x=(x^1 x^2 x^3 x^4)^T \in \mathbb R^{4}\backslash\{0\},$
where two nonzero vectors in $\mathbb R^4$ are equivalent, if one is a scalar multiple of the other. The point $p$ can be represented as $p = [x^1 : x^2: x^3 : x^4]$ (homogeneous coordinates notation). One may also represent $\mathbb RP^3$ as the sphere $S^3$ with the antipodal points identified. We will often refer to this identification as the {\em spherical representation} of the real projective space. $ \mathbb RP^3$ is an embedded manifold with the VW-embedding $j: \mathbb{R}P^3 \rightarrow \mathcal S(4,\mathbb R),$ given by
\begin{eqnarray}\label{eq:VW}
j([x]) = xx^T, x^Tx=1.
\end{eqnarray}
 Given a random object $Y$ on $\mathbb R P^3,$ $Y=[X],X^TX=1,$ such that $E(XX^T)$ has a simple largest eigenvalue, one can show that the VW (extrinsic)-mean $\mu_j =  [\gamma],$ where $\gamma $ is a unit eigenvector of $E(XX^T)$ corresponding to this largest eigenvalue (see Bhattacharya and Patrangenaru \cite{BhPa:2005}).

\noindent Our analysis will be conducted on $P\Sigma_{3}^{k}$, the  projective shape space of 3D $k$-ads in $\mathbb RP^m$ for which $\pi=([u_1], \dots, [u_5])$ is a projective frame in $\mathbb RP^3.$ $P\Sigma_{3}^{k}$ is homeomorphic to the manifold $\left( \mathbb R P^3 \right)^{k-5}$ with $k-5=q$ (see Patrangenaru et. al (2010)\cite{PaLiSu:2010}). The embedding on this space is the VW (Veronese-Whitney) embedding given by
\begin{eqnarray}\label{eq:VW-prsh}
j_k: \left( \mathbb R P^3 \right)^q \to \left( S(4, \mathbb R )\right)^q ~~~~~~~~~~~~~~~~~~~~~~~~~~~\notag\\
j_k([x_1], \dots, [x_q]) = (j([x_1]), \dots, j([x_q])),
\end{eqnarray}
with $j:\mathbb{R}P^{3}\rightarrow S_+(4,\mathbb R)$ the embedding given in \eqref{eq:VW}. Additionally $j_k$ is an equivariant embedding w.r.t. the group $(S_+(4,\mathbb R))^q$ and has the corresponding projection
\begin{eqnarray}\label{VW_Proj}
P_{j_k}: \left( S_{+}(4, \mathbb R )\right)^q \backslash  \mathcal F_q \to  j_k\left( \mathbb R P^3 \right)^q \notag\\
P_{j_k}(A_1, \dots, A_q)= \left(  j([m_1]), \dots, j[m_q])\right)
\end{eqnarray}
where $m_1, \dots, m_q$ are unit eigenvectors of $A_1, \dots, A_q$ (respectively) corresponding to the respective highest eigenvalues of those nonnegative definite symmetric matrices. Let $Y$ be be a random object from a VW distribution $Q$ on $(\mathbb{R}P^3)^{q},$ where $Y = (Y^1, \dots, Y^q),$ and $Y^s= [X^s] \in \mathbb RP^3$ for all $s = \overline{1,q}.$  The VW mean is given by
\begin{equation}\label{eq:VW-mean-q}
\mu_{j_k}=([\gamma_{1}(4)], \cdots,[\gamma_q(4)]),
\end{equation}
 where, for $ s = \overline{1,q},~ \lambda_s(r)$ and $\gamma_s(r), r=1, \dots , 4$ are the eigenvalues in increasing order and the corresponding eigenvectors of $E\left[ X^s (X^s)^T \right].$

\noindent In case of a VW-nonfocal random object $[X]$ on $\mathbb RP^3,$ we know that $\mu_{E,j}=[\nu_4],$ where $\lambda_r$ and $\nu_r,~r=1,2,3,4,$ are eigenvalues in increasing order and corresponding unit eigenvectors of $\mu=E[XX^T]$. Similarly, given i.i.d.r.o's $Y_i =[X_i], i=1,\dots, n$ from $Q$ on $\mathbb RP^3,$ their VW sample mean, is given by $\overline{X}_{E,j}= [g(4)]$, where $d(r)$ and $g(r) \in \mathbb R^4,~r=1,2,3,4,$ are eigenvalues in increasing order and corresponding unit eigenvectors of $J= \frac{1}{n} \sum_{i=1}^{n} X_{i} X_{i}^T.$

\noindent We now recall from Bhattacharya and Patrangenaru (2005) \cite{BhPa:2005} that the statistic
$$T([X],Q)=n \|S(j,X)^{-1/2} \tan_{j(\mu_{E,j})} \left( j(\overline{X}_{E,j}) - j(\mu_{E,j})\right) \|^2,$$
in case of a random sample from a distribution on $\mathbb RP^3,$  has the form  $T([X],Q)=T([X],[\nu_4])$ given by
\begin{equation}\label{Result_Original}
T([X], [\nu_4])= n g(4)^T ~[(\nu_r)]_{r=1,2,3}S(j,X)^{-1} [(\nu_r)]_{r=1,2,3}^{T} ~g(4),
\end{equation}
where the entries of the sample VW-covariance matrix are
\begin{equation}\label{WV-cov-proj}
S(j,X)_{ab}= n^{-1} (d(4)-d(a))^{-1} (d(4)- d(b))^{-1} \times \sum_{i=1}^n (g(a) \cdot X_i)(g(b) \cdot X_r)(g(4) \cdot X_i)^2,
\end{equation}
for $a,b=1,2,3.$

\noindent If we project on the tangent space to the VW-sample mean, we get the statistic
\begin{eqnarray}
T([X],\hat Q)=T([X], [g(4)])= \| S(j,X)^{-1/2} \tan_{j(\overline{X}_{E,j})} \left( j(\overline{X}_{E,j}) - j(\mu_{E,j})\right) \|^2 = \notag\\
= n ~\nu_4^T ~[g(r)]_{r=1,2,3} S (j,X)^{-1} [g(r)]_{r=1,2,3}^{T} ~\nu_4,
\end{eqnarray}
where $S(j,X)$ is also given in \eqref{WV-cov-proj}, and from the Slutsky's theorem, asymptotically $T([X], [\nu_4])$  and $T([X], [g(4)])$ both have a $\chi^{2}_{3}$ distribution (see Bhattacharya and Patrangenaru (2005) \cite{BhPa:2005}).

\noindent Before we express our statistics of interest, it will be important to note another result from { Crane and Patrangenaru (2011) \cite{CrPa:2011}} concerning the statistics
$$T(Y,\mu_{E,j_k})= n \| S_{\bar Y} (j_k, Y)^{-1/2} \tan_{j(\overline{Y}_{E,j_k})} \left( j(\overline{Y}_{E,j_k}) - j(\mu_{E,j_k})\right) \|^2$$
\noindent And this Hotelling $T^2$ type statistic is given by
\begin{equation}
T(Y, ([\gamma_{1}(4)], \cdots,[\gamma_q(4)]))= n~\left( \gamma_{1}(4)^T D_1 \dots \gamma_{q}(4)^T D_q \right)~~ S_{\bar Y}(j_k,Y)^{-1}\left( \gamma_{1}(4)^T D_1 \dots \gamma_{q}(4)^T D_q \right)^T
\end{equation}
where for $s=1,\dots,q$ we have $D_s = (g_s(1)~ g_s(2) ~g_{s}(3)) \in \mathcal M(4,3,\mathbb R)$ and for a pair of indices $(s,a), s=1,\dots,q$ and $a=1,2,3$ in their lexicographic order we have
\begin{equation}
S_{\bar Y}(j_k,Y)_{(s,a), (t,b)}=n^{-1} (d_s(4)-d_s(a))^{-1} (d_t(4)- d_t(b))^{-1} \times \sum_{i=1}^{n} (g_s(a) \cdot X^s_i)(g_t(b) \cdot X^t_i)(g_s(4) \cdot X^s_i) (g_t(4) \cdot X^t_i)
\end{equation}

\noindent In the next theorem we will take advantage of these results. { \begin{align}
\label{MANOVA_Shape}
H_0 &:\ \mu_{1,E}=\mu_{2,E}=...=\mu_{g,E}= \mu_E,\\
H_a &:\ at\ least\ one\ equality\ \mu_{a,E}=\mu_{b,E}, 1 \leq a < b\leq g \ does\ not\ hold. \notag
\end{align}
 }We aim to have an explicit representation of the expressions,
\begin{eqnarray}
T_{c}\left( Y^{(g)},\mu^{(p)}_E \right)={ n_a} \sum_{a=1}^{g} \left\| S_{\bar Y}(j_k, Y_a)^{-1/2} \tan_{j_{k}\left(\mu^{(p)}_E \right)} \left( j_{k} (\overline{Y}_{a,E}) - j_{k} \left(\mu^{(p)}_E \right) \right) \right\|^2 \\
T_{d} \left( Y^{(g)},\overline{Y}^{(p)}_E \right)={ n_a} \sum_{a=1}^{g} \left\| S_{\bar Y}(j_k, Y_a)^{-1/2} \tan_{j_{k} \left(\overline{Y}^{(p)}_E \right)} \left(j_{k}(\overline{Y}_{a,E}) - j_{k}\left(\mu^{(p)}_E \right) \right) \right\|^2
\end{eqnarray}

\noindent where $\mu_{a,E}= ([\nu_{1}^{a}(4)], \dots, [\nu_{q}^{a}(4)])$ are the VW mean from distribution $Q_a$ (of $Y_{r_a}$) and $(\eta_{s}^{a}(r), \nu_{s}^{a}(r)),$ { $r=1,\dots,4$,} are eigenvalues and corresponding unit eigenvectors of $E(X_{a,1}^{s}(X_{a,1}^{s})^T]$.  The corresponding VW sample mean is given by $\overline{Y}_{a,E}=([g_{1}^{a}(4), \dots, [g_{q}^{a}(4)]),$ where for each $s=1, \dots, q$ and $r=1,\dots, 4$, $(d_{s}^{a}(r), g_{s}^{a}(r))$ are eigenvalues in increasing order and corresponding unit eigenvectors of $J_{s}^{a}= \frac{1}{n_a} \sum_{i=1}^{n_a}X^{s}_{a,i} (X^{s}_{a,i})^T.$ Also $\mu^{(p)}_E$ is the VW pooled mean given by
\begin{eqnarray}
j_{k} \left(\mu^{(p)}_E \right)= P_{j_k}\left( \sum_{a=1}^{g} \lambda_a j_{k}(\mu_{a,E})  \right)\\
\mu^{(p)}_E=([\gamma_{1}^{(p)}(4)], \dots, [\gamma_q^{(p)}(4)]),
\end{eqnarray}
where for $s=1,\dots, q, \gamma_{1}^{(p)}(4)$ is the eigenvector corresponding to the largest eigenvalue of the $s-th$ axial component of the pooled matrix
with weights $\lambda_a, a=1,\dots, g:$
$$\sum_{a=1}^g \frac{\lambda_a}{\lambda}E(X_{a,1}X_{a,1}^T).$$

The pooled VW-sample mean $\overline{Y}^{(p)}_E$ is given by
\begin{eqnarray}
j_{k} \left(\overline{Y}^{(p)}_E \right)= P_{j_k}\left( \sum_{a=1}^{g} \frac{n_a}{n} j_{k}(\overline{Y}_{a,E})  \right)\\
\overline{Y}^{(p)}_E= ([{\bf g}_{1}^{(p)}(4)], \dots, [{\bf g}_{q}^{(p)}(4)])
\end{eqnarray}
where for $s=1,\dots,q$,  ${\bf d}_{s}^{(p)}(r)$ and ${\bf g}_{s}^{(p)}(r) \in \mathbb R^4,~r=1,2,3,4,$ are eigenvalues in increasing order and corresponding unit eigenvectors of the matrix { $J^{(p)}=\sum_{a=1}^{g} \frac{n_a}{n} j_{k}(\overline{Y}_{a,E}) .$}\\

\noindent We now express the following matrices
\begin{eqnarray}
{\bf C}_s= (\gamma_{s}^{(p)}(1)~ \gamma_{s}^{(p)}(2)~\gamma_{s}^{(p)}(3)) \in \mathcal M (4,3:\mathbb R)\\
{\bf D}_s=({\bf g}_{s}^{(p)}(1)~ {\bf g}_{s}^{(p)}(2) ~{\bf g}_{s}^{(p)}(3)) \in \mathcal M (4,3:\mathbb R)
\end{eqnarray}
{
\begin{corollary}\label{ManoveTheo2}
Assume $j_k$ is the VW embedding of $(\mathbb RP^3)^q$ and $\{Y_{a, r_a}\}_{r_a=1, \dots,n_a},~a=1, \dots,g$ are i.i.d.r. objects random from  the $j_k$-nonfocal probability measures $Q_a$ on $(\mathbb RP^m)^q,$ that have non degenerate $j_k$-extrinsic covariance matrices. Consider the statistics
\begin{enumerate}[(i)]
\item $T_{c}\left( Y^{(g)},\mu^{(p)}_E \right)= \sum_{a=1}^{g} n_a ~\left( (g_{1}^{a}(4))^T {\bf C}_1 \dots (g_{s}^{a}(4))^T {\bf C}_q \right)~S_{\bar Y_a}(j_k,Y_a)^{-1}\left(g_{1}^{a}(4)^T {\bf C}_1 \dots g_{q}^{a}(4)^T {\bf C}_q \right)^T$
\item $T_{d} \left( Y^{(g)},\overline{Y}^{(p)}_E \right)= \sum_{a=1}^{g} n_a ~\left[(  \gamma_{1}^{(p)}(4)-g_{1}^{a}(4) )^T {\bf D}_1 \dots (\gamma_{q}^{(p)}(4)-g_{q}^{a}(4) )^T {\bf D}_q \right]~~ $\\
$~~~~~~~\quad ~~\quad~~ \quad~~~~~\quad~~\quad~~~S_{\bar Y_a}(j_k,Y_{a})^{-1}$\\
$~~~~~~~\quad ~~\quad~~ \quad~~~\quad~~~~~\left[(  \gamma_{1}^{(p)}(4)-g_{1}^{a}(4) )^T {\bf D}_1 \dots (\gamma_{q}^{(p)}(4)-g_{q}^{a}(4) )^T {\bf D}_q \right]^T.$
\end{enumerate}
where
\begin{eqnarray}
 S_{\bar Y_a}(j_k,Y_a)_{(s,c)(t,b)}= n_{a}^{-1} ({\bf d}_{s}^{(p)}(4) - {\bf d}_{s}^{(p)}(c))^{-1} ({\bf d}_{t}^{(p)}(4) - {\bf d}_{t}^{(p)}(b))^{-1} \notag\\
 \times ~~~~~~~~~~~~~~~~~~~~~~~\notag\\
 ~~~~~\quad~~~~~~~~~\quad~~~~~~~~~~~~~~~\sum_{i} ({\bf g}^{(p)}_{s}(c)\cdot X_{a,i}^{s})({\bf g}^{(p)}_{t}(b) \cdot X_{a,i}^{t})({\bf g}^{(p)}_{s}(4) \cdot X_{a,i}^{s})({\bf g}^{(p)}_{t}(4) \cdot X_{a,i}^{t}) \notag
 \end{eqnarray}
  and $s,t=1, \dots,q$ and $c,b=1, \dots,m$.
If $\frac{n_a}{n} \to \lambda_a >0,$ as $n \to \infty,$ then both $T_{c}\left( Y^{(g)},\mu^{(p)}_E \right)$ and $T_{d} \left( Y^{(g)},\overline{Y}^{(p)}_E \right)$ have asymptotically a $\chi^{2}_{3q}$ distribution.
\end{corollary}}
\begin{proof}
For part $(i)$  we note that for each $a=1,\dots, g$ we get a natural extension of a result in { Bhattacharya and Bhattacharya (2012) \cite{BhBh:2012}} as shown in \eqref{Result_Original}.
For part $(ii)$ recall that
\begin{equation}
T_{d} \left( Y^{(g)},\overline{Y}^{(p)}_E \right)={ n_a} \sum_{a=1}^{g} \left\| S_{\bar Y_a}(j_k, Y_a)^{-1/2} \tan_{j_{k} \left(\overline{Y}^{(p)}_E \right)} \left(j_{k}(\overline{Y}_{a,E}) - j_{k}\left(\mu^{(p)}_E \right) \right) \right\|^2 \notag
\end{equation} we start by rewriting the expression above and we have
\begin{eqnarray}
T_{d} \left( Y^{(g)},\overline{Y}^{(p)}_E \right)={n_a} \sum_{a=1}^{g} \left\| S_{\bar Y_a}(j_k, Y_a)^{-1/2} \tan_{j_{k} \left(\overline{Y}^{(p)}_E \right)} \left(j_{k}(\overline{Y}^{(p)}_E) - j_{k}\left(\mu^{(p)}_E \right) \right) \right. \notag\\
-\left. S_{\bar Y_a}(j_k, Y_a)^{-1/2} \tan_{j_{k} \left(\overline{Y}^{(p)}_E \right)} \left(j_{k}(\overline{Y}^{(p)}_E) - j_{k}\left(\overline{Y}_{a,E} \right) \right) \right\|^2 \notag\\
T_{d} \left( Y^{(g)},\overline{Y}^{(p)}_E \right)=\sum_{a=1}^{g} n_a~ \left \| S_{\bar Y_a}(j_k, Y_a)^{-1/2}~\left[(  \gamma_{1}^{(p)}(4))^T {\bf D}_1 \dots (\gamma_{q}^{(p)}(4))^T {\bf D}_q \right]^T \right. \notag\\
-\left. S_{\bar Y_a}(j_k, Y_a)^{-1/2} \left[( g_{1}^{a}(4) )^T {\bf D}_1 \dots (g_{q}^{a}(4) )^T {\bf D}_q \right]^T \right\|^2
\end{eqnarray}
	
\end{proof}
\noindent If $Y_{a,r_a}$ are $j_k$-nonfocal distributions on $(\mathbb RP^3)^q$ with an nonzero absolutely continuous component (see Ferguson(1996)\cite{Fe:1996}, p.30), one may obtain better coverage confidence regions, using nonparameric bootstrap. Consider the pivotal statistics $T_{c}\left( Y^{(g)},\mu^{(p)}_E \right)$ and $T_{d} \left( Y^{(g)},\overline{Y}^{(p)}_E \right).$ 
under the hypothesis \\
$~~~\quad \quad \quad ~~~\begin{cases}\label{Hypo1}
H_0 &:\ \mu_{1,E}=\mu_{2,E}=...=\mu_{g,E}= \mu_E^{(p)}, \notag\\
H_a &:\ni (i,j) 1\le i<j <g, \text{s.t.} \ \mu_{i,E} \ne \mu_{j,E} . \notag
\end{cases}$\\

\begin{corollary}
The $(1-c) 100 \%$ bootstrap confidence regions for $\mu_{E}$ with $d=g p$ are given by
\begin{enumerate}[(a)]
\item ${C^*}_{n, c}^{(g)}= j^{-1}(U^{*}_{n , c})$ and
$U^{*}_{n , c}= \{  j_k(\nu) \in j_k((\mathbb RP^3)^q): T_{c}\left( Y^{(g)},\nu \right)\leq {c^*}^{(g)}_{1-c} \}$
\noindent where ${c^*}^{(g)}_{1-c}$ is the upper $100(1-c) \%$ point of the values
\begin{equation}
 T_{c}\left( {Y^*}^{(g)},\overline{Y}^{(p)}_E\right)= \sum_{a=1}^{g} n_a ~\left( ({g^*}_{1}^{a}(4))^T {\bf D}_1 \dots ({g^*}_{s}^{a}(4))^T {\bf D}_q \right)~S_{\bar Y_{a}^{*}}(j_k,Y^{*}_a)^{-1}\left({g^*}_{1}^{a}(4)^T {\bf D}_1 \dots {g^*}_{q}^{a}(4)^T {\bf D}_q \right)^T
\end{equation} among the bootstrap re samples.
\item ${D^*}_{n, c}^{(g)}= j^{-1}({V^*}_{n , c})$ and
${V^*}_{n, c}= \{ j_k(\nu) \in j_k((\mathbb RP^3)^q): T_{c}\left( Y^{(g)},\overline{Y}^{(p)}_E , \nu \right) \leq  {d^*}^{(g)}_{1-c} \}$ where \\
$T_{d} \left( Y^{(g)},\overline{Y}^{(p)}_E , \nu \right)={n_a} \sum_{a=1}^{g} \left\| S_{\bar Y_a}(j_k, Y_a)^{-1/2} \tan_{j_{k} \left(\overline{Y}^{(p)}_E \right)} \left(j_{k}(\overline{Y}_{a,E}) - j_{k}(\nu) \right) \right\|^2$
where ${d^*}^{(g)}_{1-c} $ is the upper $100(1-c)\%$ point of the values of
\begin{equation}\label{eq:manova-boot}
{T_{d} \left( {Y^*}^{(g)},\overline{Y^*}^{(p)}_E ,\overline{Y}^{(p)}_E \right)=\sum_{a=1}^{g} n_a~ \left \| S_{\bar Y_{a}^{*}}(j_k, Y^{*}_a)^{-1/2}~\tan_{j_{k} \left(\overline{Y}^{*(p)}_E \right)} \left(j_{k}(\overline{Y}^{*}_{a,E}) - j_{k}(\overline{Y}^{(p)}_E ) \right)\right\|^2}
\end{equation}
\end{enumerate}

among the bootstrap resamples. The confidence regions given by \eqref{dstar_g} and \eqref{cstar_g} have both coverage error $O_{p}(n^{-2}).$
\end{corollary}

\noindent Note that here
 \begin{eqnarray}
 S_{\bar Y_{a}^{*}}(j_k,Y^{*}_a)_{(s,c)(t,b)}= n_{a}^{-1} ({\bf d}_{s}^{*(p)}(4) - {\bf d}_{s}^{*(p)}(c))^{-1} ({\bf d}_{t}^{*(p)}(4) - {\bf d}_{t}^{*(p)}(b))^{-1} \notag\\
 \times ~~~~~~~~~~~~~~~~~~~~~~~\notag\\
 ~~~~~\quad~~~~~~~~~\quad~~~~~~~~~~~~~~~\sum_{i} ({\bf g}^{*(p)}_{s}(c)\cdot X_{a,i}^{*s})({\bf g}^{*(p)}_{t}(b) \cdot X_{a,i}^{*t})({\bf g}^{*(p)}_{s}(4) \cdot X_{a,i}^{*s})({\bf g}^{*(p)}_{t}(4) \cdot X_{a,i}^{*t}) , b, c = 1,2,3.\notag
 \end{eqnarray}

\section{Application to face data analysis}\label{sc5}
 A digital images data set was collected using a high resolution Panasonic-Lumix DMC-FZ200 camera. Our analysis will be conducted on $g=5$ individuals. The images can be found at $ani.stat.fsu.edu/\sim vic/E-MANOVA$
\noindent We tested for the existence of a 3D mean projective shape difference to differentiate between five faces which are represented in Fig \ref{Manov1}

\begin{figure}[ht!]\label{Manov1}
\begin{center}
\includegraphics[scale = 0.4]{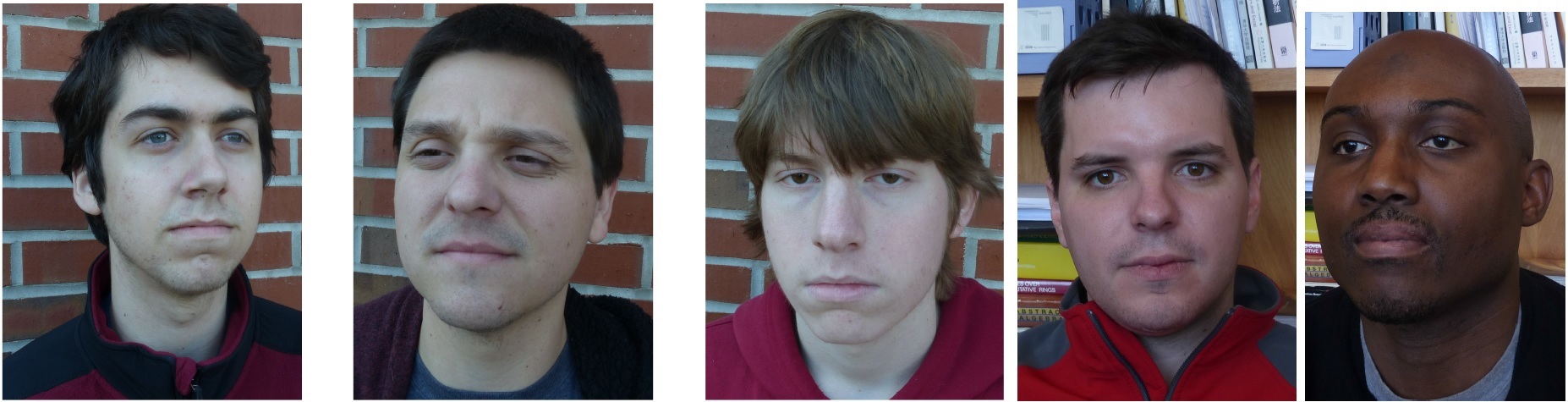}
\end{center}
\caption{\small Faces used in the extrinsic MANOVA analysis}
\end{figure}
The 3D surface reconstructions of these faces, with seven labeled landmarks, were obtained using the software Agisoft. These reconstructions (including texture) are displayed in Figure \ref{MAN_recons}.

\begin{figure}[ht!]\label{MAN_recons}
\begin{center}
\includegraphics[scale = 0.8]{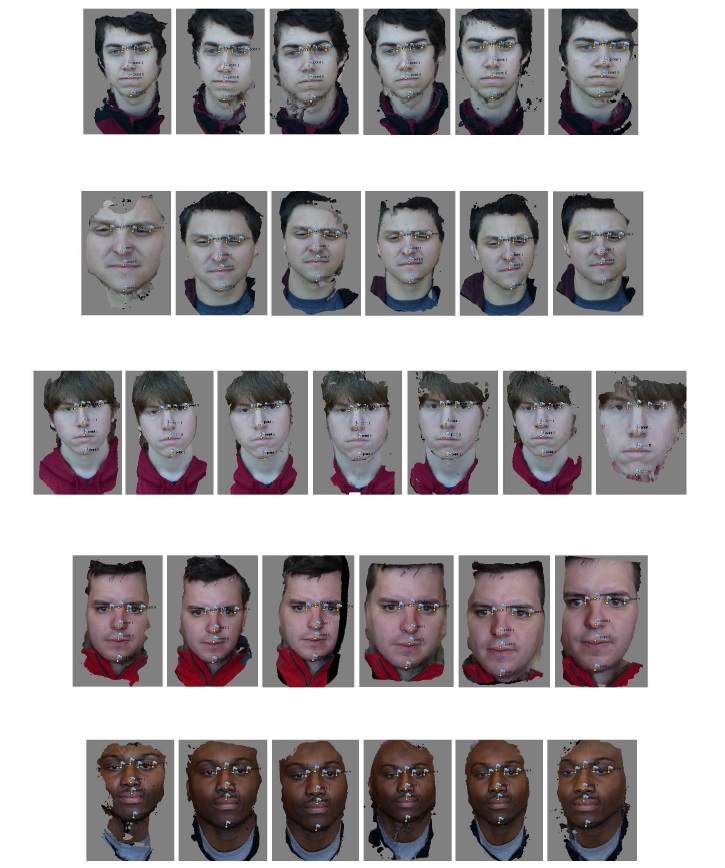}
\end{center}
\caption{\small Sample of Facial Reconstructions}
\end{figure}

\noindent The 3D reconstruction was done using the AGISOFT software. The images in Fig \ref{MAN_recons} represent 19 facial reconstructions. Each of those reconstruction was created using mostly 4 to 5  digital camera images of a given individual. We placed seven anatomical landmarks as shown across the data in Figure \ref{MAN_3}.
\begin{figure}[ht!]
\begin{center}\label{MAN_3}
\includegraphics[scale = 0.45]{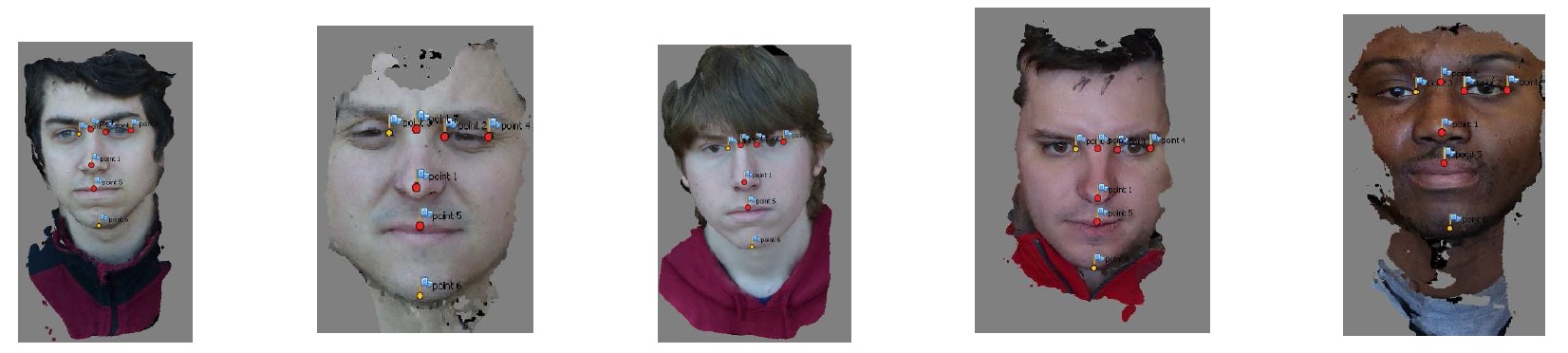}
\end{center}
\caption{\small Projective Frame Shown in Red}
\end{figure}

\noindent Five of those landmarks (colored in red) are selected as the projective frame and the resulting two projective coordinates determine the 3D projective shape of the seven landmark configuration selected. Note that we used a different projective frame than the one in Yao (2016)\cite{KDY:2016}, to insure that the landmarks are in general position. \\
We will compare these faces by conducting a MANOVA on manifold to compare $g=5$ VW-means on $P \Sigma_{3}^{7}=(\mathbb RP^3)^{2}.$ For $n=\sum_{a=1}^{5} n_a=31$ where $n_1=n_2=n4=n5=6$ and $n_3=7$ our hypothesis problem is \begin{align}
H_0 &:\ \mu_{1,E}=\mu_{2,E}=\mu_{3,E}= \mu_{4,E} =\mu_{5,E} = \mu_E,\notag\\
H_a &:\ at\ least\ one\ of \ equalities\ above\ does\ not\ hold. \notag
\end{align}
Since the true pulled mean is unknown and our data set is relatively small we will reject the null hypothesis if \\
$T_{d} \left( {Y}^{(3)},\overline{Y}^{(p)}_E  \right)=\sum_{a=1}^{5} n_a~ \left \| S_{\bar Y_{a}}(j_k, Y_a)^{-1/2}~\tan_{j_{k} \left(\overline{Y}^{(p)}_E \right)} \left(j_{k}(\overline{Y}_{a,E}) - j_{k}(\overline{Y}^{(p)}_E ) \right) \right\|^2$ is greater than ${d^*}^{(3)}_{1-\alpha},$
where ${d^*}^{(3)}_{1-\alpha}$ is the $(1-\alpha)100\%$ cutoff of the corresponding bootstrap distribution in equation \eqref{eq:manova-boot}. \\

\noindent Using $\alpha = 0.05,$ and $70543872$ resamples we obtain a value $T_{d} \left( {Y}^{(3)},\overline{Y}^{(p)}_E  \right)= 389860$ and  ${d^*}^{(3)}_{0.95}=60616,$ and we therefore reject the null hypothesis. We conclude that there exists a statistically significant VW-mean 3D-projective shape face difference between at least two of the individuals in our data set.


\begin{thebibliography}{99}

\bibitem{BhBh:2012} A. Bhattacharya, R. N. Bhattacharya (2012). {\em Nonparametric Inference  on Manifolds, with applications to Shape Spaces. } Cambridge University Press. New York, USA.

\bibitem{BhPa:2005} R. N. Bhattacharya and V. Patrangenaru (2005). Large
Sample Theory of Intrinsic and Extrinsic Sample Means
on Manifolds- Part II, {\em Annals of Statistics.} {\bf 33}, 1211--1245.

\bibitem{CrPa:2011} M. Crane and V. Patrangenaru(2011). Random Change on
a Lie Group and Mean Glaucomatous Projective Shape Change Detection From Stereo Pair Images. {\em Journal of Multivariate Analysis}. {\bf 102}, 225--237.

\bibitem{Fa:1992} O. Faugeras (1992). What can be seen in three dimensions with an uncalibrated stereo rig?. {\em Proc. European Conference on Computer Vision, LNCS.} {\bf 588}, 563--578.

\bibitem{Fe:1996} T. Ferguson.(1996) {\em A Course in Large Sample Theory}. CRC Texts in Statistical Sciences.

\bibitem{HaGuCh:1992} R. I. Hartley, R. Gupta and T. Chang(1992). Stereo from uncalibrated cameras. {\em Proceedings IEEE Conference on Computer Vision and Pattern Recognition}, 761 -- 764.

\bibitem{MaPa:2005} K. V. Mardia and {V. Patrangenaru} (2005).
Directions and Projective Shapes. {\em Annals of Statistics} {\bf
33}
1666--1699.
\bibitem{PaEl:2015} V. Patrangenaru and L. E. Ellingson(2015). {\em Nonparametric Statistics on Manifolds and their Applications to Object Data Analysis}. {CRC Texts in Statistical Science}.
\bibitem{PaLiSu:2010} V. Patrangenaru, X. Liu and S. Sughatadasa (2010). Nonparametric 3D Projective Shape Estimation from Pairs of 2D Images
- I, In Memory of W.P. Dayawansa. {\em Journal of Multivariate
Analysis}. {\bf 101}, 11--31.

\bibitem{KDY:2016} K. D. Yao(2016). {\em Data Analysis on Object Spaces and Applications}. PhD Thesis, Florida State University.

\end{thebibliography}
\end{document}